\documentclass[11pt,a4paper]{amsart}

\usepackage[utf8]{inputenc}
\usepackage[T1]{fontenc}
\usepackage{amsmath}
\usepackage{amsfonts}
\usepackage{ae}
\usepackage{units}
\usepackage{icomma}
\usepackage{color}
\usepackage{graphicx}
\usepackage{bbm}
\usepackage{amsthm}
\usepackage{amssymb}
\usepackage{cancel}
\usepackage{upgreek}
\usepackage{algorithm}
\usepackage{algorithmic}
\usepackage{type1cm}
\usepackage{eso-pic}
\usepackage[breaklinks,pdfpagelabels=false]{hyperref}
\usepackage{xspace}
\usepackage{amsrefs}

\newcommand{\abs}[1]{\ensuremath{\left| #1 \right|}}

\newcommand{\diam}{\operatorname{diam}}

\renewcommand{\epsilon}{\varepsilon}

\newtheorem{theorem}{Theorem}[section]

\newtheorem{proposition}[theorem]{Proposition}

\theoremstyle{definition}
\newtheorem{remark}[theorem]{Remark}

\newcommand{\tr}{\operatorname{Tr}}

\numberwithin{equation}{section}
\numberwithin{theorem}{section}

\begin{document}

\title{An improved energy argument for the Hegselmann-Krause model}
\author{Anders Martinsson}
\address{Department of Mathematical Sciences, Chalmers University Of Technology and University of Gothenburg, 41296 Gothenburg, Sweden}
\email{andemar@chalmers.se}
\subjclass[2010]{93A14, 39A60, 91D10}
\keywords{Hegselmann-Krause model, energy, freezing time}

\begin{abstract}
We show that the freezing time of the $d$-dimensional Hegselmann-Krause model is $O(n^4)$ where $n$ is the number of agents. This improves the best known upper bound whenever $d\geq 2$.
\end{abstract}
\maketitle

\section{Introduction}
The Hegselmann-Krause bounded confidence model, or simply the HK-model, is a simple model for opinion dynamics, first introduced in \cite{K97} and popularized in \cite{HK02}. In this model, we consider $n$ agents, indexed by integers in $[n]=\{1, 2, \dots, n\}$. Each agent $i$ initially has the opinion $\mathbf{x}_0(i)$, represented by a vector in $\mathbb{V}=\mathbb{R}^d$ for some $d\geq 1$. Two agents consider each others opinions reasonable if their Euclidean distance is at most a constant $\epsilon$, called the confidence radius. The agents update their opinions synchronously in discrete time steps by compromising with all opinions they consider reasonable. More precisely, for each $t=0, 1, \dots$ we recursively define
\begin{equation}\label{eq:HKupdate}
\mathbf{x}_{t+1}(i) = \frac{1}{\abs{\mathcal{N}_t(i)}} \sum_{j\in \mathcal{N}_t(i)} \mathbf{x}_t(j).
\end{equation}
where $\mathcal{N}_t(i) = \{j \in [n]: \| \mathbf{x}_t(i)-\mathbf{x}_t(j) \|_2 \leq \epsilon \}$. We will refer to \eqref{eq:HKupdate} as the HK update rule. In this paper we will always assume that $\epsilon = 1$.

Arguably, the most fundamental result about the HK-model is that, for any initial configuration, the system freezes after a finite number of time steps. That is, for sufficiently large $T$ we have that $\mathbf{x}_T = \mathbf{x}_{t}$ for any $t \geq T$. We will refer to the smallest such $T$ as the \emph{freezing time} of the system, and let $T_d(n)$ denote the maximal freezing time of any configuration of $n$ agents with $d$-dimensional opinions.


In \cite{C11} it is shown that, for any $d$, $T_d(n)=n^{O(n)}$ and is further conjectured that $T_d(n)$ grows polynomially in $n$. In dimension one, this was first shown in \cite{MBCF07} who proved that $T_1(n)=O(n^5)$. This has later been improved to $T_1(n)=O(n^3)$ in \cite{BBCN13}, see also \cites{TN11,MT12+}. Polynomial freezing time in arbitrary dimension was shown in \cite{BBCN13}, which obtains the bound $T_d(n) = O(n^{10} d^2)$. As opinions will always be contained in the affine space spanned by the initial opinions, we may assume that $d\leq n-1$. Hence, this implies a uniform upper bound of $T_d(n) = O(n^{12})$ independent of $d$. In a recent paper \cite{EB14+}, this was improved to $T_d(n)=O(n^8)$.

The problem of finding lower bounds on $T_d(n)$ has received less attention. In \cite{BBCN13} it was noted that it is possible to obtain freezing times of order $n^2$ for any $d\geq 2$ by placing opinions equidistantly on a circle. More recently, \cite{HW14+} shows that a certain ``dumbbell'' configuration achieves freezing time of order $n^2$ also in $d=1$.

The aim of this paper is to prove the following upper bound on the freezing time.
\begin{theorem}\label{thm:freeze}
The maximal freezing time for the $n$-agent HK model in any dimension is $O(n^4)$.
\end{theorem}

\section{Proof of Theorem 1.1}

For a given sequence $\{\mathbf{x}_t\} \in \mathbb{V}^n$, we define the \emph{communication graph} $G_t$ as the graph with vertex set $[n]$ and where $i$ and $j$ are connected by an edge if $\|\mathbf{x}_t(i)-\mathbf{x}_t(j)\|_2 \leq 1$. Note that all vertices in $G_t$ have an edge going to themselves. We will here write $i \sim_t j$ to denote that $i$ is adjacent to $j$ in $G_t$. We further let $P_t$ denote the transition matrix for a simple random walk on $G_t$. Hence, if we think of $\mathbf{x}_t$ as a $n\times d$ matrix, we can formulate the HK dynamics as
\begin{equation}
\mathbf{x}_{t+1} = P_t \mathbf{x}_t.
\end{equation}

For a configuration $\mathbf{x}=(\mathbf{x}(1), \mathbf{x}(2), \dots, \mathbf{x}(n))\in \mathbb{V}^n$ of $n$ agents, we define its energy as
\begin{equation}\label{eq:energydef}
\mathcal{E}(\mathbf{x}) = \sum_{i=1}^n\sum_{j=1}^n \min\left\{ \|\mathbf{x}(i)-\mathbf{x}(j)\|_2^2, 1\right\}.
\end{equation}
Note that the energy of any configuration lies between $0$ and $n^2$. Let $\{\mathbf{x}_t\}$ be a sequence in $\mathbb{V}^n$ which satisfies \eqref{eq:HKupdate}. This energy function has the important property that $\mathcal{E}(\mathbf{x}_t)$ is non-increasing in $t$. The following inequality has played a central roll in obtaining the upper bounds on the high-dimensional freezing time in \cites{BBCN13,EB14+}.
\begin{proposition}
\begin{equation}
\mathcal{E}(\mathbf{x}_t)-\mathcal{E}(\mathbf{x}_{t+1}) \geq 4\|\mathbf{x}_{t+1}-\mathbf{x}_t\|_2^2.
\end{equation}
\end{proposition}
\begin{proof}
See Theorem 2 in \cite{RMF08}.
\end{proof}

Here, we propose another way to estimate the energy decrement in a step in the HK-model. For a given state $\mathbf{x}$ and for any ordered pair $(i, j) \in [n]^2$, we say that $(i, j)$ is \emph{active} if $\|\mathbf{x}(i) - \mathbf{x}(j)\|_2 \leq 1$. We consequently define the active part of the energy of $\mathbf{x}$ as 
\begin{equation}
\mathcal{E}_{active}(\mathbf{x}) = \sum_{(i, j) \text{ active}} \|\mathbf{x}(i)-\mathbf{x}(j)\|_2^2
\end{equation}
\begin{proposition}\label{prop:energydec}
For each $t\geq 0$, let
\begin{equation}
\lambda_t = \max\left\{ \abs{\lambda} : \lambda\neq 1\text{ is an eigenvalue of }P_t\right\}.
\end{equation}
Then
\begin{equation}
\mathcal{E}(\mathbf{x}_t) - \mathcal{E}(\mathbf{x}_{t+1}) \geq   \left(1-\lambda_t^2\right) \mathcal{E}_{active}(\mathbf{x}_t).
\end{equation}
\end{proposition}
\begin{proof}
Let $A_t$ denote the adjacency matrix of $G_t$, and let $D_t$ denote its degree matrix, that is, the diagonal matrix whose $(i,i)$:th element is given by the degree of $i$. Recall that every vertex in $G_t$ has an edge to itself. Observe that $P_t = D_t^{-1} A_t$. We have
\begin{align*}
\mathcal{E}(\mathbf{x}_t) &= \sum_{i \sim_t j} \|\mathbf{x}_t(i)-\mathbf{x}_t(j)\|_2^2 +\sum_{i \not \sim_t j} 1\\
&= 2 \tr\left( \mathbf{x}_t^\top \left(D_t - A_t\right) \mathbf{x}_t\right) + \sum_{i \not \sim_t j} 1,
\end{align*}
where $\tr(\cdot)$ denotes trace. Here, we again interpret $\mathbf{x}_t$ as an $n\times d$ matrix. Similarly, we have
\begin{align*}
\mathcal{E}_{active}(\mathbf{x}_t) = \sum_{i \sim_t j} \|\mathbf{x}_t(i)-\mathbf{x}_t(j)\|_2^2 = 2 \tr\left( \mathbf{x}_t^\top\left(D_t - A_t\right) \mathbf{x}_t\right),
\end{align*}
and
\begin{align*}
\mathcal{E}(\mathbf{x}_{t+1}) &= \sum_{i \sim_{t+1} j} \|\mathbf{x}_{t+1}(i)-\mathbf{x}_{t+1}(j)\|_2^2 +\sum_{i \not \sim_{t+1} j} 1\\
&\leq \sum_{i \sim_{t} j} \|\mathbf{x}_{t+1}(i)-\mathbf{x}_{t+1}(j)\|_2^2 +\sum_{i \not \sim_{t} j} 1\\
&= 2 \tr\left(\mathbf{x}_{t+1}^{\top} \left(D_t - A_t\right) \mathbf{x}_{t+1}\right) + \sum_{i \not \sim_t j} 1\\
&= 2 \tr\left(\mathbf{x}_{t}^{\top} A_t D_t^{-1} \left(D_t - A_t\right) D_t^{-1} A_t \mathbf{x}_{t}\right) + \sum_{i \not \sim_t j} 1.
\end{align*}
Hence, it suffices to show that
\begin{equation}\label{eq:sufficient}
\tr\left(\mathbf{x}_{t}^{\top} A_t D_t^{-1} \left(D_t - A_t\right) D_t^{-1} A_t \mathbf{x}_{t}\right) \leq \lambda_t^2 \tr\left(\mathbf{x}_t^{\top} \left(D_t - A_t\right) \mathbf{x}_t\right).
\end{equation}

Let $\mathbf{y}_t = D_t^{1/2} \mathbf{x}_t$ and $B_t = D_t^{1/2} P_t D_t^{-1/2}=D_t^{-1/2} A_t D_t^{-1/2}$. It is straight-forward to show that \eqref{eq:sufficient} simplifies to
\begin{equation}
\tr\left(\mathbf{y}_t^{\top} B_t \left(I-B_t\right) B_t \mathbf{y}_t\right) \leq \lambda_t^2 \tr\left(\mathbf{y}_t^{\top} \left(I-B_t\right) \mathbf{y}_t\right).
\end{equation}
When $d=1$, this inequality follows by standard linear algebra: write $\mathbf{y}_t$ as a linear combination of eigenvectors of $B_t$ and observe that $B_t$ is a symmetric matrix which is similar to $P_t$. For the case when $d\geq 2$, let $e_1, \dots, e_d$ denote the standard basis of $\mathbb{V}$. We can rewrite \eqref{eq:sufficient} as
\begin{equation}
\sum_{i=1}^d \left(\mathbf{y}_t e_i\right)^{\top} B_t \left(I-B_t\right) B_t \mathbf{y}_t e_i \leq   \sum_{i=1}^d \lambda_t^2 \left(\mathbf{y}_t^{\top} e_i\right) \left(I-B_t\right) \mathbf{y}_t e_i.
\end{equation}
By the one-dimensional case, we know that this inequality holds term-wise.
\end{proof}

\begin{proposition}
For any $t\geq 0$, we have
\begin{equation}
\abs{\lambda_t} \leq 1 - \frac{1}{n^2 \diam(G_t)}
\end{equation}
where $\diam(G_t)$ denotes the graph diameter of $G_t$. If $G_t$ is not connected, we interpret $\diam(G_t)$ as the largest diameter of any connected component of $G_t$.
\end{proposition}
\begin{proof}
See for instance Corollary 13.24 in \cite{mixing}. Note that it suffices to consider the case where $G_t$ is connected.
\end{proof}
\begin{proof}[Proof of Theorem \ref{thm:freeze} ]
We call a time $t=0, 1, \dots$ a merging time if two agents with different opinions at time $t$ move to the same opinion at time $t+1$. As merges are irreversible, there can at most be $n-1$ such times.

Assume that $t<T$ is not a merging time. Then $\diam(G_t) \geq 2$. Observe that for any $i, j \in [n]$, every second edge in a minimal path from $i$ to $j$ must have length at least $\frac{1}{2}$, hence $\mathcal{E}_{active}(\mathbf{x}_t) = \Omega\left( \diam(G_t) \right)$. Applying Proposition \ref{prop:energydec}, it follows that the energy decrement in this step is $\Omega\left(\frac{1}{n^2}\right)$, and can hence occur at most $O(n^4)$ times.
\end{proof}

\begin{remark}
The energy argument presented here is optimal in the sense that there are configurations where the energy decrement is of order $\frac{1}{n^2}$. In particular, for the dumbbell in \cite{HW14+}, this is the case during the first $\Theta(n^2)$ time steps until the communication graph changes.
\end{remark}

\section*{Acknowledgements}
The author would like to thank Peter Hegarty and Edvin Wedin for helpful discussions.

\end{document}